\documentclass[conference]{IEEEtran}
\usepackage{subfig}
\usepackage{stfloats}
\IEEEoverridecommandlockouts
\newcommand{\nonl}{\renewcommand{\nl}{\let\nl\oldnl}}% Remove line number for one line
\usepackage{etoolbox}
\AtEndEnvironment{example}{\null\hfill$\square$}%
\AtEndEnvironment{remark}{\null\hfill$\square$}%

\usepackage{enumerate}
\usepackage{enumitem}
\usepackage{multirow,array}
\usepackage{cite}
\usepackage{graphicx}
\usepackage{psfrag}
\usepackage{array}
\usepackage{algorithm2e}
\usepackage{amssymb}
\usepackage{amsfonts}
\usepackage{float}
\usepackage{textcomp}
\usepackage{xcolor}
\usepackage{tabu}
\usepackage{array}
\usepackage{tabularx}
\usepackage{subfiles}

\usepackage{amssymb}% http://ctan.org/pkg/amssymb
\usepackage{pifont}% http://ctan.org/pkg/pifont
\usepackage[amsthm,thmmarks]{ntheorem}

\newcolumntype{P}[1]{>{\centering\arraybackslash}p{#1}}
\newcolumntype{M}[1]{>{\centering\arraybackslash}m{#1}}
\usepackage{multirow}
\usepackage{booktabs}

\usepackage{amsmath}
\usepackage{float}
\newtheorem{proposition}{Proposition}
\newtheorem{conjecture*}{Conjecture}

\newtheorem{corollary}{Corollary}

\newtheorem{remark}{Remark}
\newtheorem{definition}{Definition}
\newtheorem{theorem}{Theorem}

\newcommand{\pkarxiv}[1]{}
\newcommand{\mileak}{\rho^{\mathsf{MIL}}}
\newcommand{\maxleak}{\rho^{\mathsf{MaxL}}}

\newcommand{\Prob}{\mathsf{P}}

\newcommand{\comment}[1]{}

\newcommand{\graph}{\mathcal{G}}
\newcommand{\vertset}{\mathcal{V}}
\newcommand{\edgeset}{\mathcal{E}}
\newcommand{\indset}[1]{\mathcal{I}_{#1}}
\newcommand{\Wtheta}{W_{\theta}}
\newcommand{\query}[1]{Q_{#1}}
\newcommand{\querydown}[1]{Q_{#1}^{\downarrow}}
\newcommand{\queryup}[1]{Q_{#1}^{\uparrow}}
\newcommand{\degdown}{d_n^{\downarrow}}
\newcommand{\degup}{d_n^{\uparrow}}
\newcommand{\degn}{d_n}
\newcommand{\Un}{U_n}
\newcommand{\F}{\mathbb{F}_2}
\newcommand{\xor}{\oplus}
\newcommand{\hw}[1]{w(#1)}
\newcommand{\Hb}{H_{\mathrm{b}}}
\newcommand{\ebk}{\mathbf{e}_k}
\newcommand{\KN}{K_N}
\DeclareMathOperator{\MI}{I}  % mutual information

\title{Weak Private Information Retrieval\\ for Graph-based Storage} 

\begin{document}
%%%%%%%%%%%%%%%%%%%%%%%%%%%%%%%%%%%%%%%%%%%%%%%%%%%%%%%%%%%%
% \author{
%   \IEEEauthorblockN{Authors}

% \vspace{-0.2cm}
% %%%%%%%%%%%%%%%%%%%
% }

\author{Shodasakshari Vidya, Chandan Anand, Prasad Krishnan % <-this % stops a space
%\thanks{\hrule}%

\thanks{Shodasakshari, Chandan, and Dr. Krishnan are with the Signal Processing and Communications Research Center, International Institute of Information Technology, Hyderabad, 500032, India (email: $\{$shodasakshari.vidya@research., chandan.anand@research., prasad.krishnan@$\}$iiit.ac.in). 
%Acknowledgment: \chandan{Acknowledge support details.}\pk{This is done only after acceptance}
% Shodasakshari Vidya, Chandan Anand and  Dr. Krishnan acknowledge support from ANRF-SERB project CRG/2023/008696 and ANRF/ARGM/2025/000480/TS. Dr. Kurri acknowledges support from ANRF project ANRF/ECRG/2024/005472/ENS. 
}
}

\maketitle

\allowdisplaybreaks

\begin{abstract}
  A  distributed storage system with graph-based replication consists of a collection of databases and the files they contain. The databases (or servers) are represented as the vertices of a graph, while each file is stored in a distinct pair of servers and is represented by an edge of this graph. Private information retrieval (G-PIR) on such a graph-based storage system involves a client which seeks to retrieve a desired file via a query-response protocol, without leaking the identity of the desired file index to any database. The goal of G-PIR is to maximize the rate (reciprocal of the total normalized download) under the privacy constraint. Prior work on G-PIR has involved perfect information-theoretic privacy (i.e., null leakage). However, if the privacy constraint is relaxed, then PIR protocols could be designed that have even higher rates. We term such protocols as Graph-based Weak Private Information Retrieval (G-WPIR) protocols and initiate their formal study in this work. We propose a G-WPIR scheme for arbitrary graphs, and identify the trade-offs it achieves between rate and privacy, under two well known leakage metrics: mutual information leakage and maximal leakage. Our protocol employs minimal subpacketization (representing a file-size constraint) and employs a simple probabilistic query realization to obtain the smooth trade-off. We extend this protocol with some modifications to two special classes of graphs, the complete graphs and the complete bipartite graphs, and identify the corresponding rate-privacy trade-offs achieved. 
% This work initiates the study of  file stored exactly at two databases and relax the privacy constraint in a controlled manner. We propose two achievability schemes: $1)$ a G-WPIR scheme for general graphs and $2)$ a G-WPIR scheme for complete and complete bipartite graphs. Additionally, we establish their rate-privacy trade-offs. In the G-WPIR scheme for complete graphs, we proposed a technique to improve the rate-privacy trade-offs.
\end{abstract}
% {\small Due to space considerations, this submitted manuscript is a short version of our extended manuscript \cite{Vidya_ISITA2026_arxiv}, which contains missing proofs, examples, and further results for the special graph classes.}

\section{Introduction}
In Private Information Retrieval (PIR), a client desires to download a file from a set of distributed servers while keeping the identity of the desired file private from each server. To retrieve the client's desired file privately, the client uses a query-response protocol, which ensures such private retrieval is called a PIR protocol. A PIR protocol downloads extra file bits to mask the desired file bits, thus achieving the desired privacy. The efficiency of a PIR protocol is measured using its \textit{rate}, which is the ratio of the size of the desired file to the total number of downloaded bits in the protocol. For a given system model, the \textit{capacity} of PIR is the supremum of the rate over all PIR protocols. A large body of work exists on PIR, with a considerable focus in the last several years in the information-theoretic privacy setting, starting from the seminal work by Sun and Jafar for non-colluding and colluding replicated databases \cite{sun2017capacity,sun2017_T_capacity}, including coded storage \cite{banawan2018capacity}, and low subpacketization (or file-size) schemes \cite{tian2019capacity, zhou2020capacity}.
% Some interesting extensions for the classical PIR problem are investiged in evesdroppers/Byzantine servers \chandan{Secure pir from colluding databases with eavesdroppers, The capacity of PIR with eavesdroppers, On PIR and symmetric PIR from colluding databases with adversaries and eavesdroppers, The capacity of private information retrieval from Byzantine and colluding databases}

% Consider a system model with $K$ files that are stored in $N$ servers in a distributed manner. If these servers do not communicate with each other, they are called the collusion-free servers; otherwise, they are called colluding servers. PIR protocols can be classified into two categories based on per-server download cost: fixed-download PIR protocols, where the number of downloaded bits is constant across all servers, and variable-download PIR protocols, where the number of downloaded bits may vary across servers. However, for a system model, both fixed-download scheme \cite{sun2017capacity,banawan2018capacity,sun2017_T_capacity}, and variable-download scheme \cite{tian2019capacity, zhou2020capacity} achieve the capacity.

Any distributed storage can be modeled as (hyper)graph-based storage, where the servers are represented as vertices and files stored across any two or more servers are represented as edges or hyperedges, respectively. PIR for such graph-based storage (which we term as G-PIR) was first studied in \cite{ravivetal_2018_TIT}, in which achievable schemes as well as converse results were presented, thus establishing lower and upper bounds on the capacity of G-PIR for graph replication colluding and non-colluding setups. G-PIR for storage based on simple graphs (where each file is stored on a distinct set of two servers, and no two servers share more than a single file) was further enhanced in \cite{Sadehetal_TIT_2023_BoundsonPIRgraphs,ge2025privateinformationretrievalgraphs,Kong_TIT_2025newcapacityboundspir,Yuhangetal_Capof4StarGraph_ISIT_2023}. In particular, these works presented results for G-PIR on various special classes of graphs such as star graphs, bipartite graphs, path graphs, and complete graphs. Any protocol for the complete graph $K_N$ on $N$ vertices can be extended to a protocol for an arbitrary $N$-vertex graph, and thus the lower bounds for $K_N$ generalize to all graphs. Further extensions include G-PIR for multi-graph-based storage \cite{meel2025privateinformationretrievalmultigraphbased} in which two servers can share more than a single file, as well as symmetric G-PIR schemes \cite{meel2025effectcommonrandomnessreplication}. 

Recently, a new class of G-PIR schemes for general graphs and star-graphs was introduced in \cite{shanbhag2026private}. The G-PIR scheme for general graphs presented in that work \cite[Section IV] {shanbhag2026private} employs a probabilistic query generation protocol, requiring minimal \textit{subpacketization} $L=1$ (number of parts that the files must be subdivided into, for protocol execution). The setup of the protocol requires a pre-processing step, which is a decomposition of the storage graph into independent sets. For a graph with independence number $\alpha(\graph)$, this protocol achieves a rate lower bounded by $1/(N-\alpha(\graph)/2)$ for any graph $\graph$ with $N$ vertices (databases) and independence number $\alpha(\graph)$, and exactly $1/(N-1)$ for the complete graph $K_N$. This protocol achieves the best-known rate for balanced bipartite graphs, while higher-rate protocols are known \cite{ge2025privateinformationretrievalgraphs} for complete graphs.

% \subsection{Special graphs and State-of-the-art-bounds}
% Specific graph structures have been extensively studied in the literature.  For a complete graph denoted as $\KN$, where each server shares exactly one file with another server, so a total of $N(N-1)/2$ files. Initially an achievable scheme provided by \cite{Sadehetal_TIT_2023_BoundsonPIRgraphs,kong2025newcapacityboundspir} where as state-of-the-art upper bound is presented in \cite{ge2025privateinformationretrievalgraphs} which is $R\le\frac{1}{\sum_{i=2}^N\frac{1}{i!}}\frac{1}{N}$.

% For a complete bipartite graph $K_{N_1, N_2}$ where $N_1$ and $N_2$ are the number of vertices of two subsets, where we consider $N_1>N_2$ and each vertices of a subset have an edge with each vertex of another subset. The achievable rate for this graph is given by \cite{kong2025newcapacityboundspir,ge2025privateinformationretrievalgraphs} and the best upper bound is presented in \cite{shanbhag2026private}, which is rate $R \geq \frac{4}{3N}$.

Relaxing the perfect privacy constraint in PIR may enable achieving larger rates, albeit with some information leakage at the servers. Such schemes are termed weak-PIR (WPIR) schemes. The study of WPIR was initiated in multiple lines of work \cite{lin2021multi,jia2019capacity,samy2021asymmetric}. In these and subsequent works, various leakage metrics were used to measure the privacy leakage about the desired file index, including: mutual information leakage \cite{lin2021multi}, maximal leakage \cite{issa2019operational}, worst-case information leakage \cite{kopf2007information}, differential privacy \cite{samy2021asymmetric}, and converse induced privacy metric \cite{CIPM2024}. These leakage measures have previously been adopted in various WPIR settings, such as colluding-free replicated servers \cite{Anandetal_ISIT2025_WPIR, anand2026converseboundssunjafartypeweak}, MDS-coded storage servers \cite{orvedal2024weaklyprivateinformationretrievalmdscoded, Anandetal_ISIT2025_WPIR, anand2026converseboundssunjafartypeweak}, and colluding servers \cite{Anandetal_ISIT2025_WPIR, anand2026converseboundssunjafartypeweak}. 

In this work, we propose a new WPIR scheme for graph-based storage, adopting both mutual information leakage and maximal leakage as our privacy metrics. The contributions of this work are as follows:
\begin{itemize}[leftmargin=*]
    \item Section~\ref{sec:model} contains preliminaries, system model, and G-WPIR basics. Subsection \ref{subsec:main_results} contains the statements of the main results of this work. Theorem \ref{thm:mi_rate} describes the rate-privacy tradeoff for arbitrary graphs, with respect to the mutual information leakage metric. Theorem \ref{thm:maxl_rate} gives the rate-privacy tradeoffs under the maximal leakage metric. 
    \item In Section \ref{sec:wpirgraphs}, we present our G-WPIR scheme, which achieves the rate-privacy tradeoffs mentioned in Theorems \ref{thm:mi_rate} and \ref{thm:maxl_rate}. This protocol is a natural weaker version of the general graph-PIR scheme proposed in \cite[Section IV]{shanbhag2026private}. We show the Theorems \ref{thm:mi_rate} and \ref{thm:maxl_rate} by obtaining the rate achieved by this protocol, and the corresponding leakage metrics.  
    \item In Section~\ref{sec:special_graphs}, we consider two special classes of graphs: complete and complete bipartite graphs. We construct G-WPIR schemes for these, using the protocol in Section \ref{sec:wpirgraphs} and its refinements. We show the corresponding rate-privacy tradeoffs for these graphs. 
\end{itemize}
In summary, this work initiates the formal study of WPIR for graph-based storage, extending results in prior work for WPIR in other models of distributed storage \cite{zhouetal_2020_WPIR_MaxL, lin2021multi, Anandetal_ISIT2025_WPIR}. We end the paper in Section \ref{sec:conclusion} with directions for further work. 

\textit{Notation}: Let $N$ be a positive integer then $[N]$ denotes $\{1,2,\cdots,N\}$. Consider $\indset{}\subset[N]$ and a set $\mathbf{a}=(a_1, a_2, \cdots, a_n)$ then we define $\mathbf{a}_{\indset{}}=\{a_i:i\in\indset{}\}$. For disjoint sets $A$ and $B$, their disjoint union is denoted as $A\sqcup B$. The symbols $\boldsymbol{0}$ and $\boldsymbol{1}$ denote all-zero and all-one vectors over $\F$ of appropriate lengths, understood from the context. For $q\in\F^{n}$, let $\hw{q}$ denote its Hamming weight. The symbol $\xor$ denotes the logical XOR operator. A biased coin toss is modeled by a Bernoulli random variable $X\sim\mathsf{Bern}(p)$, where $P(X=1)=1-P(X=0)=p$. The case $p=1/2$ corresponds to a fair coin. Let $\log$ denote the base-$2$ logarithm. The binary entropy function is denoted by $H_b(p)=-p\log(p)-(1-p)\log(1-p)$. For a random variable $X$, $\mathbb{E}[X]$ denotes expectation with respect to its underlying distribution of $X$. An (undirected) graph $\graph$ is composed of vertices ${\cal V}$ and edges ${\cal E}$. An independent set of a graph is a collection of vertices, no two of which share an edge. The size of the largest independent set of a graph is termed its independence number, denoted by $\alpha(\graph)$.

\section{System model and Main Results}
\label{sec:model}
% ----------------------------------------------------------
\subsection{System Model}
\label{subsec:system_model}

\subsubsection{Storage Model}

We consider a distributed storage system consisting of $N$ servers, indexed by the vertex set $\vertset = [N] \triangleq \{1, 2, \ldots, N\}$ of a simple, undirected graph $\graph = (\vertset, \edgeset)$. The edge set $\edgeset$ defines the replication structure, where each edge $\{i, j\} \in \edgeset$ represents a distinct file $W_{i,j}$ stored jointly on servers $i$ and $j$. Server $n$ stores all files $\{W_{n,m} : \{n,m\} \in \edgeset\}$, which we also write as $W_n^1, W_n^2, \ldots, W_n^{\degn}$, where $\degn = |\{m : \{n,m\} \in \edgeset\}|$ denotes the degree of vertex $n$. 
% \pk{PROBLEMATIC NOTATION BECAUSE it is $e_j$ and not $e_{j,n}$ MOVING HERE FROM NOTATION SECTION WHERE IT IS INAPPROPPRIATE TO HAVE. MAYBE TO RESOLVE LATER: We also use the symbol $e_j$ to refer to the edge in $\edgeset$ corresponding to the $j$-th upstream edge incident to server $n$.} 
We denote the total number of files by $K = |\edgeset|$. All files are assumed to be independent and uniformly distributed, each of entropy $H(W_{i,j}) = L$ bits, where $L \geq 1$ denotes the \emph{subpacketization} of the scheme. Throughout this work, we restrict attention to the unit-subpacketization setting, i.e., $L = 1$.
%-------------------------------------------------------------------
\subsubsection{PIR Protocol}
\label{subsec:pir_protocol} 
 
A client wishes to privately retrieve a target file $\Wtheta$, where the file index $\theta$ is drawn uniformly at random from the edge set $\edgeset$. The client communicates with the servers via a \emph{query-response} protocol: it sends a query $\query{n}$ to each server $n \in \vertset$, and receives a corresponding response $A_n$. The query $\query{n}$ is a function of $\theta$ and private client-side randomness, independent of the stored files. 

%\smallskip
% \noindent\textbf{Server responses.}
% Server $n$ stores all files $\{W_{n,m} : \{n,m\} \in \edgeset\}$, which we also write as $W_n^1, W_n^2, \ldots, W_n^{\degn}$, where $\degn = |\{m : \{n,m\} \in \edgeset\}|$ denotes the degree of vertex $n$. 
The query to server $n$ takes the form of a binary vector $\query{n} = (f_1^n, f_2^n, \ldots, f_{\degn}^n) \in \F^{\degn}$. Upon receiving a non-zero query, server $n$ returns the linear combination 
\begin{equation}
    A_n \;=\; \bigoplus_{j=1}^{\degn} f_j^n \, W_n^j,
    \label{eq:server_response}
\end{equation}
where $\xor$ denotes addition over $\F$. The query $\query{n} = \mathbf{0}$ is interpreted as a null query, eliciting no response. 

%\smallskip
% \noindent\textbf{Correctness.}
A PIR protocol must recover the client's desired file
% is \emph{correct} if the client can recover 
$\Wtheta$ from the
received responses, i.e.,
\begin{equation}
    H\!\left(\Wtheta \mid \query{[N]},\, A_{[N]}\right) = 0.
    \label{eq:correctness}
\end{equation}
% \\
% \noindent\textbf{Rate.}
The \emph{rate} of a variable-download PIR protocol is defined as
\begin{equation}
    R \;\triangleq\; \frac{L}{\mathbb{E}[D]},
    \label{eq:rate}
\end{equation}
where $L = 1$ is the subpacketization and $\mathbb{E}[D] = \sum_{n \in \vertset} H(A_n \mid Q_n)$ is the expected total download. The \emph{capacity} $C(\graph)$ is the supremum of $R$ over all correct PIR protocols.
%-------------------------------------------------------------------
\subsection{Strict and Weak Privacy}
\label{subsec:privacy}

The classical (strict, or information-theoretically perfect) privacy requirement in PIR demands that no individual server learns any information about the desired file index $\theta$ from its query-response pair:
\begin{equation}
    \MI(\theta;\, \query{n}, A_n) = 0, \quad \forall\, n \in \vertset.
    \label{eq:strict_privacy}
\end{equation}

% \pk{TO INCLUDE REFERENCES, Chandan can help perhaps: Information-theoretically perfect PIR protocols for graph-based storage was initiated in \cite{} and further developed for various general and specific graph classes in a sequence of works \cite{}.}\chandan{Will search and add more references for G-PIR variants.} 
Information-theoretic perfect PIR protocols for graph-based storage were initiated in \cite{ravivetal_2018_TIT} and further developed for various general and specific graph classes in a sequence of works \cite{Sadehetal_TIT_2023_BoundsonPIRgraphs, Kong_TIT_2025newcapacityboundspir, Yuhangetal_Capof4StarGraph_ISIT_2023, ge2025privateinformationretrievalgraphs}.
In this work, we study \emph{weak} (or $\varepsilon$-private) PIR protocols, in which each server is permitted to learn a bounded but nonzero amount of information about $\theta$. We adopt two 
% complementary \pk{May 11: its not clear to me what exactly you want to say by the word `complementary'.} 
leakage metrics to quantify this.
 
\begin{definition}[Per-server Mutual Information Leakage]
\label{def:mi_leakage}
The \emph{mutual information (MI) leakage} at server $n$ is defined as
\begin{equation}
    \varepsilon_n^{\mathrm{MI}} \;\triangleq\; \MI\!\left(\theta;\, \query{n}, A_n\right).
    \label{eq:mi_leakage}
\end{equation}
A protocol is \emph{$\varepsilon$-MI-private} if $\varepsilon_n^{\mathrm{MI}} \leq \varepsilon$ for all $n \in \vertset$.
\end{definition}
 
\begin{definition}[Per-server Maximal Leakage]
\label{def:max_leakage}
The \emph{maximal leakage} \cite{issa2019operational} at server $n$ is defined as
\begin{equation}
    \varepsilon_n^{\mathrm{max}} \;\triangleq\; \log \sum_{q_n\in\mathcal{Q}_n}
        \max_{m \in [K]}\, P_{Q_n \mid \theta}(q_n \mid m),
    \label{eq:max_leakage}
\end{equation}
where the sum ranges over all realizations $q_n$ of $\query{n}$. A protocol is \emph{$\varepsilon$-maximally-private} if $\varepsilon_n^{\mathrm{max}} \leq \varepsilon$ for all $n \in \vertset$.
\end{definition}
% In addition to mutual information leakage and maximal leakage, other leakage measures, differential privacy, and worst-case information leakage were used to quantify leakages in various WPIR settings.
These leakage measures have been adopted in various WPIR settings before, such as %. For instance, these leakages have been used in 
non-colluding and replicated servers \cite{Huangetal_ISIT2024_WPIR_HetTrustServers, zhao2025optimizingDP, anand2026converseboundssunjafartypeweak}, for MDS-coded storage servers \cite{orvedal2024weaklyprivateinformationretrievalmdscoded, anand2026converseboundssunjafartypeweak} and colluding servers \cite{anand2026converseboundssunjafartypeweak}. In this work, we adopt these metrics to present new protocols for PIR with graph-based storage.
%\pk{COMPLETE THIS SENTENCE, Chandan can help: These leakage measures have been adopted in various PIR settings before, such as non-colluding servers \cite{}, server-collusion \cite{},etc. In this work, we adopt these metrics to present new protocols for PIR with graph-based storage.}  
%--------------------------------------------------------------------
%--------------------------------------------------------------------
\subsection{Main Results}
\label{subsec:main_results}
In this subsection, we present the main results of this work, which new G-WPIR achievable rate-privacy tradeoffs with respect to the mutual information and maximal leakage metrics, for arbitrary graph-based storage. The new scheme which achieves these tradeoffs is a natural extension of the scheme for complete graphs available in \cite{shanbhag2026private}[Section IV]. Towards explaining our protocol, we first explain the pre-processing step used in \cite{shanbhag2026private} that obtains a partition of the graph into independent sets.  
\subsubsection{Preliminaries}
\label{subsec:preliminaryindepsetsandquerynotation}
% \noindent\textbf{Graph Partition:}
%\label{subsec:partition}
Following \cite{shanbhag2026private}, our protocols rely on a partition of the vertex set $\vertset$ into $\kappa$ disjoint independent sets $\indset{1}, \indset{2}, \ldots, \indset{\kappa}$, satisfying $\vertset = \bigsqcup_{s=1}^{\kappa} \indset{s}$. 
% \pk{To be added to Notation: For disjoint sets $A,B$, their disjoint union is denoted as $A\sqcup B$}\chandan{Done}. 
This partition is constructed sequentially: for each $s \in [\kappa]$, the set
$\indset{s}$ is taken as a largest  
independent set of the induced subgraph $\graph[\indset{s} \cup \cdots \cup \indset{\kappa}]$. We refer to this as the \emph{sequential independent-set partition}\footnote{The execution of the protocol does not depend on the assumption of $\indset{s}$ being the \textit{largest} independent set at that stage. It is sufficient to consider maximal independent sets at each step instead. We can construct such a decomposition in polynomial time via a greedy algorithm. The rate achieved will depend on the size of these independent sets; specifically, the quantity $\alpha(\graph)$ appearing in~\eqref{eq:download_exact} and~\eqref{eq:rate_lb} below is $|\indset{1}|$ in general, and equals $\alpha(\graph)$ precisely when $\indset{1}$ is chosen as a largest independent set of $\graph$, as in the sequential construction above.} of $\graph$.

% \noindent\textbf{Query Structure:}
For a server $n \in \indset{s}$, the edges incident on $n$ are classified according to the partition. The \emph{downstream degree} $\degdown$ of vertex $n$ is the number of edges
connecting $n$ to a vertex in $\indset{s+1} \cup \cdots \cup \indset{\kappa}$, while the \emph{upstream degree} $\degup = \degn - \degdown$ is the number of edges connecting $n$ to a vertex in $\indset{1} \cup \cdots \cup \indset{s-1}$. Note that $\degup \geq 1$ for all $n \notin \indset{1}$, since otherwise $n$ could have been absorbed into an earlier independent set, contradicting its maximality. 
% \pk{This observation may be removed in shorter version if its not used}. 
For each upstream edge $e_j$ incident on server $n$, we write $e_j \in \edgeset$ for the corresponding file index, where $j \in [\degup]$.
 
Accordingly, the query vector $\query{n}$ is decomposed as
\begin{equation}
    \query{n} \;=\; \left(\querydown{n},\, \queryup{n}\right),
    \label{eq:query_decomp}
\end{equation}
where $\querydown{n} \in \F^{\degdown}$ contains the weights for downstream edges and $\queryup{n} \in \F^{\degup}$ contains the weights for upstream edges. For servers in $\indset{1}$, the upstream subvector is  non-existent, and $\query{n}$ consists entirely of $\querydown{n}$.

\subsubsection{Leakage-Rate Tradeoffs}
The proofs of the following theorems are given in Section~\ref{sec:wpirgraphs}.

\begin{theorem}[MI Leakage and Rate Tradeoff]
\label{thm:mi_rate}
% \pk{May 12: Introduce $U_n$ here or previously}
According to the independent-set decomposition $\vertset = \bigsqcup_{s=1}^{\kappa} \indset{s}$, the per-server MI leakage satisfies
\begin{equation}
    \varepsilon_n^{\mathrm{MI}} \;=\;
    \begin{cases}
        0, & n \in \indset{1}, \\[4pt]
        H\!\left(\queryup{n}\right) - \Un\cdot\Hb(p),
            & n \in \indset{s},\; s \geq 2,
    \end{cases}
    \label{eq:mi_leakage}
\end{equation}
where
\begin{equation}
    H\!\left(\queryup{n}\right) = -\!\sum_{q \in \F^{\Un}}
    P\!\left(\queryup{n} = q\right) \log P\!\left(\queryup{n} = q\right),
    \label{eq:H_up}
\end{equation}
with the marginal $P(\queryup{n} = q)$ given by~\eqref{eq:P0}--\eqref{eq:marginal}. For servers in $\indset{\kappa}$, the formula holds with $\Un = \degn$ (since $\degdown = 0$). Also, the biased protocol with parameter $p \in [1/2, 1]$ achieves rate $R(p) \triangleq 1/\mathbb{E}[D(p)]$, where the expected download is
\begin{align}
    \mathbb{E}[D(p)]
    &= |\indset{1}|\cdot(1-p)
    \notag \\
    &\;\;+ \sum_{s=2}^{\kappa-1}\sum_{n \in \indset{s}}
      \!\!\left[1 - \frac{p^{\Un+1}}{K}\!\left(K -
      \frac{\Un(2p-1)}{p}\right)\right]
    \notag \\
    &\;\;+ \sum_{n \in \indset{\kappa}}
      \!\!\left[1 - \frac{p^{\Un}}{K}\!\left(K -
      \frac{\Un(2p-1)}{p}\right)\right].
    \label{eq:download_exact}
\end{align}
This expression holds for any sequential independent-set partition satisfying the degree properties of Subsection~\ref{subsec:preliminaryindepsetsandquerynotation}, whether each $\indset{s}$ is chosen as a largest or merely a maximal independent set (Footnote~1). Under the primary construction, in which $\indset{1}$ is itself a largest independent set of $\graph$, $|\indset{1}| = \alpha(\graph)$, recovering the form used in the remainder of the paper.
\end{theorem}

\begin{theorem}[Maximal Leakage and Rate Tradeoff]
\label{thm:maxl_rate}
The biased protocol with parameter $p \in [1/2, 1]$ achieves the same rate $R(p)$ as given by~\eqref{eq:download_exact}. According to the independent-set decomposition, the per-server maximal leakage satisfies
\begin{equation}
    \varepsilon_n^{\mathrm{max}} \;=\;
    \begin{cases}
        0, & n \in \indset{1}, \\[4pt]
        \displaystyle\log\!\left[\frac{p + p^{\Un}(1-2p)}{1-p}\right],
            & n \notin \indset{1}.
    \end{cases}
    \label{eq:max_leakage_closed}
\end{equation}
\end{theorem}
%-----------------------------------------------------------------------
%-----------------------------------------------------------------------
%-----------------------------------------------------------------------
%-----------------------------------------------------------------------
\section{A Weak PIR Scheme for General Graphs}
\label{sec:wpirgraphs}
We now present a weak PIR protocol for general graphs. Our protocol is an extension of the protocol from \cite{shanbhag2026private} which achieves perfect information-theoretic privacy. 

\subsection{A G-WPIR Scheme for arbitrary graph}
\label{subsec:scheme}
The strictly private protocol of~\cite{shanbhag2026private} generates queries sequentially over the independent sets, and our G-WPIR scheme here follows the same. We summarize Each server $n\in\indset{1}$ is assigned a binary query $\query{n}\in\{\mathbf{0},\mathbf{1}\}$ according to a $\mathrm{Bern}(p)$ distribution. For $n\in\indset{s}$, $s\geq2$, the downstream subvector $\querydown{n}$ is drawn from $\{\mathbf{0},\mathbf{1}\}$ via an independent $\mathrm{Bern}(p)$ coin, while each upstream bit is set according to a rule that involves the corresponding upstream bit and the desired file index. Choosing the value of $p$ allows for a smooth trade-off between the rate and the privacy, thus resulting in Theorems \ref{thm:mi_rate} and \ref{thm:maxl_rate}. Fixing $p=0.5$ recovers the protocol in \cite{shanbhag2026private}[Section IV], which is strictly private. An example of the execution of the protocol is available in \cite{shanbhag2026private}[Example 2] for the $p=0.5$ case. This example also illustrates the execution of our WPIR protocol here, except that in our case $p$ is any value in $[0.5,1]$.
% This achieves strict privacy with $\mathbb{E}[D]\leq N-\alpha(\graph)/2$ and $R\geq\max\!\left\{\tfrac{2}{2N-\alpha(\graph)},\tfrac{1}{N-1}\right\}$.

% The central idea of the present work is to replace the fair coin with a \emph{biased} one parametrized by $p\in[1/2,1]$: each downstream subvector is now drawn with $P(\querydown{n}=\mathbf{0})=p$, independently across servers, while the upstream rule~\eqref{eq:baseprotocol_upstream} is unchanged. Setting $p=1/2$ recovers the strictly private scheme exactly; for $p>1/2$, null queries become more probable, reducing the expected download and raising the rate at the cost of nonzero leakage at every server with at least one upstream edge. The bias parameter $p$ thus controls the leakage-rate tradeoff studied in the remainder of this section. 
We now present the complete design of the protocol. Let $\graph = (\vertset, \edgeset)$ be the storage graph on $N$ vertices and $K = |\edgeset|$ files, with sequential independent-set partition $\vertset = \indset{1} \sqcup \indset{2} \sqcup \cdots \sqcup \indset{\kappa}$ as defined in Subsection~\ref{subsec:preliminaryindepsetsandquerynotation}. Let $\theta \in \edgeset$ denote the desired file index, drawn uniformly from $\edgeset$, and fix a bias parameter $p \in [1/2,\, 1]$. The client generates a binary query vector $\query{n} \in \F^{\degn}$ for each server $n \in \vertset$ by the following sequential procedure. 

\smallskip
\noindent\textbf{Step 1 -- Queries for $\indset{1}$.}
Every server $n \in \indset{1}$ has no upstream edges ($\degup = 0$), so all incident edges are downstream ($\degdown = \degn$). For each $n\in\indset{1}$, the client draws an independent coin $C_n$ with $P(C_n = 0) = p$ and $P(C_n = 1) = 1-p$, and sets
\begin{equation}
    \query{n} =
    \begin{cases}
        \mathbf{1} \in \F^{\degn}, & \text{with probability } 1-p, \\[2pt]
        \mathbf{0} \in \F^{\degn}, & \text{with probability } p.
    \end{cases}
    \label{eq:step1}
\end{equation}

\smallskip
\noindent\textbf{Step $s$ -- Queries for $\indset{s}$, $s \in \{2,\ldots,\kappa\}$.}
For each $n \in \indset{s}$, let $\Un \triangleq \degup = \degn - \degdown$ denote the number of upstream edges, and let $\mathcal{U}_n = \{e_1, e_2, \ldots, e_{\Un}\} \subseteq \edgeset$ be the collection of upstream edges (equivalently, the file indices) at server $n$. For each $j \in [\Un]$, the file $e_j$ is shared with some upstream server $m_j \in \indset{t}$, $t < s$. Let $X_j \triangleq C_{m_j}$ be the coin drawn for $m_j$ in Step $t$; since $e_j$ is a downstream edge of $m_j$, the value $X_j$ is precisely the query
bit that $m_j$ assigned to the shared file $W_{e_j}$. The variables $X_1, \ldots, X_{\Un}$ are mutually independent with $P(X_j = 0) = p$ and $P(X_j = 1) = 1-p$ for all $j$. The query $\query{n} = (\querydown{n},\, \queryup{n})$ is constructed as follows.
\begin{enumerate}
    \item \noindent\emph{Downstream subvector $\querydown{n}$:}
            If $\degdown > 0$, the client draws a fresh coin $C_n$ with $P(C_n = 0) = p$ and $P(C_n = 1) = 1-p$, independent of all prior randomness and of $\theta$, and sets
            \begin{equation}
                \querydown{n} =
                \begin{cases}
                    \mathbf{1} \in \F^{\degdown}, & \text{with probability } 1-p, \\[2pt]
                    \mathbf{0} \in \F^{\degdown}, & \text{with probability } p.
                \end{cases}
                \label{eq:downstream}
            \end{equation}
            If $\degdown = 0$ (i.e., $n \in \indset{\kappa}$), the downstream subvector is absent.
    \item \noindent\emph{Upstream subvector $\queryup{n}$:}
            The $j$-th component of $\queryup{n}$, corresponding to the upstream file $e_j \in \mathcal{U}_n$, is set to
            \begin{equation}
                f_j^n = X_j \;\xor\; \mathbf{1}[\theta=e_j], \quad j \in [\Un].
                \label{eq:upstream}
            \end{equation}
            That is, $f_j^n = X_j$ when $e_j \neq \theta$ (the indicator is zero, bits pass unmodified), and $f_j^n = X_j \xor 1$ when $e_j = \theta$ (the indicator is one, the bit is flipped).
\end{enumerate}

\smallskip
\noindent\textbf{Server responses.}
Server $n$ returns $A_n = \bigoplus_{j=1}^{\degn} f_j^n W_n^j$ if $\query{n} \neq \mathbf{0}$, and sends no response if $\query{n} = \mathbf{0}$. The client recovers the desired file as $\Wtheta = \bigoplus_{n \in \vertset} A_n$.

\begin{remark}
\label{rem:boundary}
% \pk{can be removed in shorter version}
At $p = 1/2$ all coins are fair, every query bit is $\mathrm{Bern}(1/2)$ independently of $\theta$, and the protocol coincides with the strictly private scheme of~\cite{shanbhag2026private}. The central idea of the present work is to replace the fair coin with a \emph{biased} one parametrized by $p\in[1/2,1]$. Setting $p>1/2$ makes null queries become more probable, reducing the expected download and raising the rate at the cost of nonzero leakage at every server with at least one upstream edge. The bias parameter $p$ thus controls the leakage-rate tradeoff, as indicated by Theorems \ref{thm:mi_rate} and \ref{thm:maxl_rate}. At $p = 1$, every coin is $0$ almost surely, so $X_j = 0$ for all $j$ and rule~\eqref{eq:upstream} reduces to $f_j^n = \mathbf{1}[\theta = e_j]$, thus completely revealing $\theta$ to the downstream server storing $\Wtheta$. 
% For $p \in (1/2,1)$, the protocol is weakly private with leakage strictly between these extremes, as characterized below.
\end{remark}

\begin{proposition}[Correctness]
\label{prop:correctness}
The protocol satisfies $\bigoplus_{n \in \vertset} A_n = \Wtheta$.
\end{proposition}
 
\begin{proof}
We repeat the arguments in \cite{shanbhag2026private} for completeness. Consider a file $W_{m,n}$, which is the $i^{\text{th}}$ file at upstream server $m$ and the $j^{\text{th}}$ file at the downstream server $n$. Suppose $W_{m,n} \neq \Wtheta$, then by the query construction $f_i^m=f_j^n$. However, if $W_{m,n} =\Wtheta$, then $f_i^m=f_j^n\oplus 1$. Thus, the sum of all responses $\bigoplus_n A_n=W_\theta$. This completes the proof. 

% For any file $W_{i,j} \neq \Wtheta$, the indicator in~\eqref{eq:upstream} is zero at both endpoints, so both assign the same bit value to the file, and its contribution to the XOR sum is zero. For $\Wtheta = W_{i,j}$, exactly one endpoint (the one in the later independent set) applies the flip, so the contribution is $1$. Hence $\bigoplus_n A_n = \Wtheta$.
\end{proof}
%---------------------------------------------------------------------- 
\section{Completing The Proofs of Theorems 1 and 2}
This section establishes the three components needed to prove Theorems~\ref{thm:mi_rate} and~\ref{thm:maxl_rate}.
Section~\ref{subsec:rate_biased} derives the expected download $\mathbb{E}[D(p)]$~\eqref{eq:download_exact}, which gives the achievable rate $R(p)$ common to both theorems. Section~\ref{subsec:mi_leakage} computes the per-server MI leakage, establishing~\eqref{eq:mi_leakage} of Theorem~\ref{thm:mi_rate}. 
Section~\ref{subsec:max_leakage} computes the per-server maximal leakage, establishing~\eqref{eq:max_leakage_closed} of Theorem~\ref{thm:maxl_rate}.
%%%
\subsection{Rate}
\label{subsec:rate_biased} 
With $L = 1$, the rate is $R = 1/\mathbb{E}[D]$ where $\mathbb{E}[D] = \sum_{n \in \vertset} P(\query{n} \neq \mathbf{0})$, as the download from each server is $1$ bit.

\smallskip
\noindent\textbf{Servers in $\indset{1}$.} From~\eqref{eq:step1}, $P(\query{n} = \mathbf{0}) = p$, so $P(\query{n} \neq \mathbf{0}) = 1-p$.

\smallskip
\noindent\textbf{Servers in $\indset{s}$, $s \geq 2$.} Since $\querydown{n} \perp \queryup{n}$, the null-query event $\{\query{n} = \mathbf{0}\}$ requires both subvectors to be zero. We first compute $P(\queryup{n} = \mathbf{0})$ by conditioning on $\theta$:
\begin{equation}
    P\!\left(\queryup{n} = \mathbf{0}\right)
    = \frac{1}{K}\sum_{t=1}^{K}
    P\!\left(\queryup{n} = \mathbf{0} \mid \theta = t\right).
    \label{eq:null_up_expand}
\end{equation}
Recall that each upstream vertex of vertex $n$ has precisely one file shared with $n$. Hence, for the $K - \Un$ indices $t \notin \mathcal{U}_n$, all bits in $\queryup{n}$ are distributed as $\mathrm{Bern}(1-p)$. Further, note that  all upstream vertices at $n$ use independent coin flips to generate their respective downstream query subvectors. Thus, we have $P(\queryup{n} = \mathbf{0} \mid \theta = t) = p^{\Un}$. For each of the $\Un$ indices $t = e_k \in \mathcal{U}_n$, bit $k$ follows $\mathrm{Bern}(p)$ while the other $\Un - 1$ bits follow $\mathrm{Bern}(1-p)$, so $P(\queryup{n} = \mathbf{0} \mid \theta = e_k) = (1-p) \cdot p^{\Un-1}$. Substituting into~\eqref{eq:null_up_expand}:
\begin{align}
    P\!\left(\queryup{n} = \mathbf{0}\right)
    &= \frac{1}{K}\!\left[(K - \Un)\,p^{\Un}
       + \Un\,(1-p)\,p^{\Un-1}\right]
    \notag\\
    &= \frac{p^{\Un-1}}{K}\!\left[Kp + \Un(1-2p)\right].
    \label{eq:null_up}
\end{align}
For $n \in \indset{s}$ with $\degdown > 0$ ($s \in [2:\kappa-1]$), multiplying by $P(\querydown{n} = \mathbf{0}) = p$, we have
\begin{equation}
    P\!\left(\query{n} = \mathbf{0}\right)
    = \frac{p^{\Un+1}}{K}\!\left[K - \frac{\Un(2p-1)}{p}\right].
    \label{eq:null_mid}
\end{equation}
For $n \in \indset{\kappa}$ ($\degdown = 0$), the full query equals the upstream subvector:
\begin{equation}
    P\!\left(\query{n} = \mathbf{0}\right)
    = \frac{p^{\Un}}{K}\!\left[K - \frac{\Un(2p-1)}{p}\right].
    \label{eq:null_last}
\end{equation}

Since $(2p-1)/p \in [0,\,1]$ for $p\in[1/2,1]$ and $\Un \leq K$, the term $\left(K - \Un(2p-1)/p\right)$ is non-negative. Hence the null-query probabilities in~\eqref{eq:null_mid} and~\eqref{eq:null_last} are non-negative, and from~\eqref{eq:download_exact} we have the following corollary.

\begin{corollary}[Rate lower bound]
\label{cor:rate_lb}
Therefore,
\begin{equation}
    R(p) \;\geq\; \max\!\left\{
        \frac{1}{N - |\indset{1}|\cdot p},\;
        \frac{1}{N-1}
    \right\}.
    \label{eq:rate_lb}
\end{equation}
Under the primary (largest-independent-set) construction, $|\indset{1}|=\alpha(\graph)$ and~\eqref{eq:rate_lb} reads $R(p)\ge\max\{1/(N-\alpha(\graph)p),\,1/(N-1)\}$; setting $p=1/2$ recovers $R \geq 2/(2N - \alpha(\graph))$ from~\cite{shanbhag2026private}. Since $|\indset{1}|\leq\alpha(\graph)$ for any merely maximal choice of $\indset{1}$, and~\eqref{eq:rate_lb} is increasing in $|\indset{1}|$, the largest-independent-set construction gives the best rate guarantee attainable from this bound; a greedy maximal-only choice can only weaken it, consistent with Footnote~1. For the complete graph $K_N$, every independent set is a singleton, so $|\indset{1}|=\alpha(K_N)=1$ regardless of which vertex is chosen, and $\mathbb{E}[D(p)] = N - 1$ up to correction terms, giving $R(p) \approx 1/(N-1)$.
\end{corollary}

This establishes~\eqref{eq:download_exact} and the rate $R(p)=1/\mathbb{E}[D(p)]$ claimed in both Theorems~\ref{thm:mi_rate} and~\ref{thm:maxl_rate}.
%--------------------------------------------------------------------
\subsection{Mutual Information (MI) Leakage}
\label{subsec:mi_leakage}
We now derive the per-server MI leakage, deriving~\eqref{eq:mi_leakage} of Theorem~\ref{thm:mi_rate}.
Since $A_n$ is a deterministic function of $\query{n}$ and the stored files, and the files are independent of $\theta$, we have $I(\theta;\,\query{n}, A_n) = I(\theta;\,\query{n})$. We now find $I(\theta;\,\query{n})$ at each server $n$.

\smallskip
\noindent\textbf{Servers in $\indset{1}$.}
As $\query{n} \perp \theta$ for every $n \in \indset{1}$ by query design, we have $I(\theta;\,\query{n}) = 0$.

\smallskip
\noindent\textbf{Servers in $\indset{s}$, $s \geq 2$.}
Since $\querydown{n}$ is generated by the independent coin $C_n$, we have $\querydown{n} \perp \theta$ and $\querydown{n} \perp \queryup{n}$.
By the chain rule of mutual information,
\begin{equation}
    I(\theta;\,\query{n})
    = I(\theta;\,\querydown{n}) + I(\theta;\,\queryup{n} \mid \querydown{n})
    = I(\theta;\,\queryup{n}),
    \label{eq:chain_rule}
\end{equation}
so the entire leakage is carried by the upstream subvector. We compute $I(\theta;\,\queryup{n}) = H(\queryup{n}) - H(\queryup{n} \mid \theta)$.

\smallskip
\noindent\emph{Conditional entropy $H(\queryup{n} \mid \theta)$.}
From~\eqref{eq:upstream}, bit $j$ is $f_j^n = X_j \xor \mathbf{1}[e_j = \theta]$. The coins $X_1,\ldots,X_{\Un}$ are mutually independent with $P(X_j = 0) = p$.

\emph{Case A ($\theta \notin \mathcal{U}_n$):} The indicator is zero for all $j$, so $f_j^n = X_j \sim \mathrm{Bern}(p)$ for all $j$. The bits are i.i.d., giving $H(\queryup{n} \mid \theta) = \Un \cdot \Hb(p)$.

\emph{Case B ($\theta = e_k$ for some $k \in [\Un]$):} Only bit $k$ is flipped: $f_k^n = X_k \xor 1 \sim \mathrm{Bern}(1-p)$, while $f_j^n = X_j \sim \mathrm{Bern}(p)$ for all $j \neq k$. Since $\Hb(p) = \Hb(1-p)$, and all $\Un$ bits remain jointly independent, thus $H(\queryup{n} \mid \theta = e_k) = \Un \cdot \Hb(p)$.

%\emph{Case A ($\theta \notin \mathcal{U}_n$):} The indicator is zero for all $j$, so $f_j^n = X_j \sim \mathrm{Bern}(1-p)$ for all $j$. The bits are i.i.d., giving $H(\queryup{n} \mid \theta) = \Un \cdot \Hb(p)$.

%\emph{Case B ($\theta = e_k$ for some $k \in [\Un]$):} Only bit $k$ is flipped: $f_k^n = X_k \xor 1 \sim \mathrm{Bern}(p)$, while $f_j^n = X_j \sim \mathrm{Bern}(1-p)$ for all $j \neq k$. Since $\Hb(p) = \Hb(1-p)$, and all $\Un$ bits remain jointly independent, thus $H(\queryup{n} \mid \theta = e_k) = \Un \cdot \Hb(p)$.

In both cases $H(\queryup{n} \mid \theta) = \Un \cdot \Hb(p)$, and averaging over $\theta$ confirms
\begin{equation}
    H\!\left(\queryup{n} \mid \theta\right) = \Un \cdot \Hb(p).
    \label{eq:cond_entropy}
\end{equation}

\smallskip
\noindent\emph{Marginal distribution of $\queryup{n}$.}
We compute $P(\queryup{n} = q)$ for each $q \in \F^{\Un}$ by the law of total probability. Define the \emph{no-hit} conditional distribution
\begin{equation}
    P_0(q) \;\triangleq\; P\!\left(\queryup{n} = q \mid \theta \notin
    \mathcal{U}_n\right) \;=\; p^{\Un - \hw{q}}\,(1-p)^{\hw{q}},
    \label{eq:P0}
\end{equation}
and for each $k \in [\Un]$ the \emph{hit-at-$k$} conditional distribution 
\begin{equation}
    P_k(q) \;\triangleq\; P\!\left(\queryup{n} = q \mid \theta = e_k\right)
    \;=\; p^{\Un - \hw{q \xor \ebk}}\,(1-p)^{\hw{q \xor \ebk}},
    \label{eq:Pk}
\end{equation}
where $\ebk \in \F^{\Un}$ is the standard basis vector with $1$ at position $k$, reflecting the flip of bit $k$ in~\eqref{eq:upstream}. Noting that
\begin{equation}
    \hw{q \xor \ebk} =
    \begin{cases}
        \hw{q} - 1, & q_k = 1,\\
        \hw{q} + 1, & q_k = 0,
    \end{cases}
    \label{eq:hw_flip}
\end{equation}
and using $P(\theta \notin \mathcal{U}_n) = (K - \Un)/K$ and $P(\theta = e_k) = 1/K$, the marginal distribution of $\queryup{n}$ is
\begin{equation}
    P\!\left(\queryup{n} = q\right)
    \;=\; \frac{K-\Un}{K}\,P_0(q)
      \;+\; \frac{1}{K}\sum_{k=1}^{\Un} P_k(q).
    \label{eq:marginal}
\end{equation}
Combining~\eqref{eq:chain_rule} and~\eqref{eq:cond_entropy}: $I(\theta;\,\query{n}) = H(\queryup{n}) - \Un \cdot \Hb(p)$, with $H(\queryup{n})$ computed from the marginal~\eqref{eq:marginal}
via~\eqref{eq:H_up}.

\begin{remark}
At $p = 1/2$, both $P_0(q) = 2^{-\Un}$ and $P_k(q) = 2^{-\Un}$ for all $q$ and $k$, so the marginal~\eqref{eq:marginal} is the uniform distribution on $\F^{\Un}$, giving $H(\queryup{n}) = \Un$ bits and $\varepsilon_n^{\mathrm{MI}} = \Un - \Un = 0$, confirming strict privacy. Note that the MI leakage is zero at servers in
$\indset{1}$.
\end{remark}
Combined with the rate result of Section~\ref{subsec:rate_biased}, this completes the proof of Theorem~\ref{thm:mi_rate}.
%----------------------------------------------------------------------
\subsection{Maximal Leakage}
\label{subsec:max_leakage}
We next derive the per-server maximal leakage, establishing~\eqref{eq:max_leakage_closed} of Theorem~\ref{thm:maxl_rate}.
Recall that the maximal leakage at server $n$ is defined as $\varepsilon_n^{\mathrm{max}} = \log \sum_{q \in \mathcal{Q}_n} \max_{m \in [K]} P_{Q_n \mid \theta}(q \mid m)$~\cite{issa2019operational}. %\pk{the following line can be removed in short version} 
Operationally, for each possible observed query $q$, the server identifies the single file index $m$ that maximises the probability of having generated $q$; these maximal probabilities are then summed over all possible queries and the logarithm is taken. 

We now compute the maximal leakages for the biased protocol.
%\pk{One line argument needed for maximal leakage = 0 at $n\in\indset{1}$. I switched the statement to keep that part first.}
\noindent\textit{\underline{Servers IN $\indset{1}$:}}
Since $\query{n}\perp\theta$ for every $n\in\indset{1}$, we have $\max_{\theta}P_{Q_n|\theta}(q\mid\theta)=P_{Q_n}(q)$ for each $q$, so $\sum_{q}\max_{\theta}P_{Q_n|\theta}(q\mid\theta)=1$ and $\varepsilon_n^{\mathrm{max}}=\log 1=0$.

\noindent\textit{\underline{Servers in $\indset{s}$, $s\geq2$:}}
Let $q=(q^\downarrow,q^\uparrow)$ denote a query sample at server $n$, with the upstream and downstream sub-queries $q^\uparrow$ and $q^\downarrow$ respectively. denote the  $\querydown{n}$ is independent of $\theta$ and the full query factorises as $P_{Q_n \mid \theta}(q \mid m) = P_{\querydown{n}}(q^\downarrow) \cdot P_{\queryup{n} \mid \theta}(q^\uparrow \mid m)$, summing $P_{\querydown{n}}(q^\downarrow)$ over all $q^\downarrow \in \F^{\degdown}$ gives $1$.
The downstream subvector therefore vanishes from the maximal leakage sum, giving
\begin{equation}
    \varepsilon_n^{\mathrm{max}}
    = \log \sum_{q^{\uparrow} \in \F^{\Un}}
      \max_{m \in [K]}\, P\!\left(\queryup{n} = q^{\uparrow} \mid \theta = m\right).
    \label{eq:max_red}
\end{equation}
We identify the maximising $m$ for each $q^\uparrow$ by expressing $P_k(q^\uparrow)$ as a scalar multiple of $P_0(q^\uparrow)$. From~\eqref{eq:P0} and~\eqref{eq:Pk}, when $\theta = e_k \in \mathcal{U}_n$ only the $k$-th bit is flipped, so
\begin{align}
    \frac{P_k(q^\uparrow)}{P_0(q^\uparrow)}
    &\;=\;
    \frac{p^{\Un - \hw{q^\uparrow \xor \ebk}}\,(1-p)^{\hw{q^\uparrow \xor \ebk}}}
         {p^{\Un - \hw{q^\uparrow}}\,(1-p)^{\hw{q^\uparrow}}}\nonumber\\
    &\;=\;
    \begin{cases}
        \dfrac{p}{1-p}, & q^\uparrow_k = 1,\\[6pt]
        \dfrac{1-p}{p}, & q^\uparrow_k = 0,
    \end{cases}
    \label{eq:ratio}
\end{align}
using~\eqref{eq:hw_flip}. Since $p \in [1/2, 1)$ implies $p/(1-p) \geq 1$ while $(1-p)/p \leq 1$, the ratio in~\eqref{eq:ratio} exceeds $1$ if and only if $q^\uparrow_k = 1$. We now determine the maximising $m$ case by case.
 
\emph{Case I: $\hw{q^\uparrow} = 0$.} All bits of $q^\uparrow$ are $0$, so for every $k \in [\Un]$ the ratio $P_k(q^\uparrow)/P_0(q^\uparrow) = (1-p)/p \leq 1$. No upstream hypothesis $\theta = e_k$ can exceed $P_0(q^\uparrow)$, so $\theta \notin \mathcal{U}_n$ is the maximiser:
\begin{equation}
    \max_{m}\, P\!\left(\queryup{n} = q^{\uparrow} \mid \theta = m\right)
    \;=\; P_0(\mathbf{0}) \;=\; p^{\Un}.
    \label{eq:max_w0}
\end{equation}
 
\emph{Case II: $\hw{q^\uparrow} \geq 1$.} There exists at least one position $k$ with $q^\uparrow_k = 1$. For such $k$, the ratio $P_k(q^\uparrow)/P_0(q^\uparrow) = p/(1-p) \geq 1$, strictly exceeding $1$ for $p > 1/2$. The server can therefore select any such $k$ as the maximising hypothesis, giving
\begin{equation}
    \max_{m}\, P\!\left(\queryup{n} = q^{\uparrow} \mid \theta = m\right)
    \;=\; \frac{p}{1-p} \cdot P_0(q^{\uparrow}).
    \label{eq:max_w1}
\end{equation}
 
Summing over all $2^{\Un}$ possible query vectors and grouping by Hamming weight:
\begin{align}
    &\sum_{q^{\uparrow} \in \F^{\Un}}
    \max_{m}\, P\!\left(\queryup{n} = q^{\uparrow} \mid \theta = m\right)
    \notag\\
    &\quad= p^{\Un}
      + \frac{p}{1-p}
        \sum_{w=1}^{\Un} \tbinom{\Un}{w}\,p^{\Un-w}(1-p)^{w}.
    \label{eq:sum_expanded}
\end{align}
By the binomial theorem, $\sum_{w=0}^{\Un}\binom{\Un}{w}p^{\Un-w}(1-p)^w = (p + 1-p)^{\Un} = 1$, so $\sum_{w=1}^{\Un}\binom{\Un}{w}p^{\Un-w}(1-p)^w = 1 - p^{\Un}$.
Substituting into~\eqref{eq:sum_expanded}:
\begin{align}
    \sum_{q^{\uparrow}} \max_m\, P\!\left(\queryup{n} = q^{\uparrow} \mid \theta = m\right)
    &= p^{\Un} + \frac{p}{1-p}\!\left(1 - p^{\Un}\right)
    \notag\\
    &= \frac{p^{\Un}(1-p) + p - p^{\Un+1}}{1-p}
    \notag\\
    &= \frac{p + p^{\Un}(1-2p)}{1-p}.
    \label{eq:sum_max}
\end{align}
Taking logarithms gives~\eqref{eq:max_leakage_closed}.

\begin{remark}
\label{rem:max_properties}
For fixed $p\in(1/2,1)$, the leakage is strictly increasing in $\Un$: differentiating the numerator with respect to $\Un$ gives $p^{\Un}\ln(p)(1-2p)>0$ (since $\ln(p)<0$ and $1-2p<0$ for $p\in(1/2,1)$), so a server with more upstream edges always incurs higher maximal leakage. The maximal leakage depends on the graph only through the upstream degree $\Un$, and is also strictly increasing in $p$ for fixed $\Un\geq1$.
\end{remark}
Combined with the rate result of Section~\ref{subsec:rate_biased}, this completes the proof of Theorem~\ref{thm:maxl_rate}.
\section{Special Graphs}
\label{sec:special_graphs}
 
We now specialize the general scheme of Section~\ref{sec:wpirgraphs} and resulting theorems for two important graph families: complete graphs and complete bipartite graphs. For each family, the simplicity of the graph-structure allows the general leakage and rate expressions to be substantially simplified. The symmetry of complete graphs allows us to add a further randomization during query generation, which enables us to get uniform leakages at all servers. 
%%%
\subsection{Complete Graphs}
\label{subsec:complete_graphs}
 
Consider the complete graph $\KN$ on $N$ vertices, in which every pair of servers shares a file. The total number of files is $K = \binom{N}{2} = N(N-1)/2$, and the independence number is $\alpha(\KN) = 1$, since no two vertices are non-adjacent. The sequential independent-set partition of Subsection~\ref{subsec:preliminaryindepsetsandquerynotation} therefore consists of $\kappa = N$ singleton sets, one vertex per set: $\indset{s} = \{v_s\}$ for each $s \in [N]$, where the vertices are labeled in the order they are selected. The server $v_s \in \indset{s}$ has total degree $\degn = N-1$, of which $\degdown = N-s$ edges point downstream (to servers in $\indset{s+1}, \ldots, \indset{N}$) and $\degup = s - 1$ edges point upstream (to servers in $\indset{1}, \ldots, \indset{s-1}$). In particular, the server in $\indset{1}$ has no upstream edges ($\degup = 0$), and the server in $\indset{N}$ has no downstream edges ($\degdown = 0$).

\begin{theorem}[Leakage-Rate Tradeoff for $\KN$]
\label{thm:kn_tradeoff}
The biased protocol on $\KN$ with parameter $p\in[1/2,1]$ satisfies:

\noindent\textit{(i) MI leakage:} For server $v_s$,
\begin{equation}
    \varepsilon_{v_s}^{\mathrm{MI}}
    \;=\; H\!\left(Q_{v_s}^{\uparrow}\right) - (s-1)\,\Hb(p),
    \label{eq:mi_complete}
\end{equation}
where $H(Q_{v_s}^\uparrow)$ is computed from~\eqref{eq:P0}--\eqref{eq:marginal}
with $\Un=s-1$ and $K=N(N-1)/2$. The MI leakage is zero for $s=1$ and
for all servers at $p=1/2$, and increases with $p$ for each fixed
$s\geq2$; the per-server variation with $s$ is illustrated in
Fig.~\ref{fig:comlete_graph_MIL}.

\noindent\textit{(ii) Maximal leakage:}
\begin{equation}
    \varepsilon_{v_s}^{\mathrm{max}} \;=\;
    \begin{cases}
        0, & s=1,\\[4pt]
        \displaystyle\log\!\left[\frac{p+p^{s-1}(1-2p)}{1-p}\right],
            & s\in\{2,\ldots,N\}.
    \end{cases}
    \label{eq:max_complete}
\end{equation}
The leakage~\eqref{eq:max_complete} is strictly increasing in $s$:
servers encountered later accumulate more upstream edges and therefore
leak more, with maximum leakage at $\indset{N}$.

\noindent\textit{(iii) Rate:} The expected download is
\begin{align}
    \mathbb{E}\!\left[D_{\KN}(p)\right]
    &= (1-p) \notag\\
    &\quad+\sum_{s=2}^{N-1}\!\!\left[1-\frac{p^{s}}{K}\!\left(K
      -\frac{(s-1)(2p-1)}{p}\right)\right] \notag\\
    &\quad+\left[1-\frac{p^{N-1}}{K}\!\left(K
      -\frac{(N-1)(2p-1)}{p}\right)\right],
    \label{eq:download_complete}
\end{align}
with $K=N(N-1)/2$, giving rate $R_{\KN}(p)=1/\mathbb{E}[D_{\KN}(p)]$.
\end{theorem}

\begin{proof}
Parts~(i) and~(ii) follow by substituting $\Un=s-1$ into
Theorems~\ref{thm:mi_rate} and~\ref{thm:maxl_rate} respectively.
Part~(iii) follows by substituting $|\indset{s}|=1$ and $\alpha(\KN)=1$
into~\eqref{eq:download_exact}.
\end{proof}

\begin{remark}
\label{rem:complete_boundary}
The expected download~\eqref{eq:download_complete} satisfies
$\mathbb{E}[D_{\KN}(1/2)]=N-1$ and $\mathbb{E}[D_{\KN}(1)]=1$,
giving $R_{\KN}(1/2)=1/(N-1)$ and $R_{\KN}(1)=1$.
At $p=1/2$: $(2p-1)=0$, so null-probabilities reduce to $(1/2)^s$
and $(1/2)^{N-1}$; using $\sum_{s=2}^{N-1}(1/2)^s = 1/2-(1/2)^{N-1}$
gives $\mathbb{E}[D_{\KN}(1/2)]=N-1$, recovering~\cite{shanbhag2026private}.
At $p=1$: the non-null probability for server $v_s$ is $(s-1)/K$, so
$\mathbb{E}[D_{\KN}(1)]=(1/K)\sum_{s=1}^{N-1}s=1$.
At $p=1/2$ the protocol is strictly private and all $N-1$ non-first
servers respond. At $p=1$ only the server whose upstream edge coincides
with $\theta$ responds, giving $D=1$ and $R=1$ — the maximum rate at
the cost of completely revealing $\theta$. The rate $R_{\KN}(p)$ is
strictly increasing in $p$, ranging from $1/(N-1)$ to $1$.
\end{remark}

\begin{remark}
\label{rem:boundarycompletegraphs}
The boundary values have a natural interpretation. At $p = 1/2$, the protocol is strictly private and the download equals $N-1$: essentially, all $N-1$ servers other than the one in $\indset{1}$ (which is queried only half the time) respond. At $p = 1$, the biased coins are all identically zero, so only the server whose upstream edge coincides with the desired file $\theta$ receives a non-null query on average; precisely one server responds, giving $D = 1$ and $R = 1$.  This represents the maximum rate achievable by the protocol family, at the cost of revealing $\theta$ entirely to a server which contains it. The rate $R_{\KN}(p)$ is strictly increasing in $p$, ranging from $1/(N-1)$ at $p=1/2$ to $1$ at $p=1$.
\end{remark}
%-------------------------------------------------------------------------
\subsection{Leakage-Averaged Protocol for $\KN$ via Cyclic Shift}
\label{subsec:cyclic_shift}
 
In the biased protocol of Section~\ref{subsec:scheme}, the leakage at server $v_s$ strictly increases with its partition index $s$: the server in $\indset{1}$ leaks nothing, while the server in $\indset{N}$ incurs the maximum leakage. This asymmetry is purely an artefact of the ordering imposed by the sequential partition; the underlying graph $\KN$ treats all vertices identically. We now describe a modified protocol that equalises the per-server leakage across all servers by prepending a uniform random cyclic relabelling of the vertices before the query process, while leaving the expected download — and hence the rate — unchanged.
 
\smallskip
\noindent\textbf{Protocol 2 (Cyclic-Shift Biased Protocol).}
Before executing the biased protocol of Section~\ref{subsec:scheme}, the client draws a shift offset $R$ uniformly and independently at random from $[N]$. It relabels the vertices of $\KN$ by the cyclic permutation $\sigma_R : v_i \mapsto v_{(i + R - 2 \bmod N) + 1}$, and runs the biased protocol on the relabelled graph. The offset $R$ is private client-side randomness, not revealed to any server. Server $v_n$ therefore observes only its own query $Q_n$, which is distributed according to the mixture
\begin{equation}
    P\!\left(Q_n = q \mid \theta\right)
    \;=\; \frac{1}{N} \sum_{s=1}^{N}
          P_s\!\left(Q = q \mid \theta\right),
    \label{eq:mix_dist}
\end{equation}
where $P_s(Q = q \mid \theta)$ is the query distribution of a server at position $s$ under the original biased protocol.
 
\begin{theorem}[Protocol~2 Leakage-Rate Tradeoff] \label{thm:shift_leakage}
Under Protocol~2 with parameter $p\in[1/2,1)$, every server $v_n$ has the same marginal query distribution and hence identical leakage by any measure. The expected download satisfies $\mathbb{E}[D_{\mathrm{shift}}]=\mathbb{E}[D_{\KN}(p)]$, where $\mathbb{E}[D_{\KN}(p)]$ is given by~\eqref{eq:download_complete}, while the common per-server leakage values $\varepsilon_{\mathrm{shift}}^{\mathrm{max}}$ (maximal leakage) and $\varepsilon_{\mathrm{shift}}^{\mathrm{MI}}$ (MI leakage) satisfy:

\noindent\textit{(i) Maximal leakage upper bound:}
\begin{align}
    \varepsilon_{\mathrm{shift}}^{\mathrm{max}}
    &\leq\log\!\left(\frac{1}{N}\sum_{s=1}^{N}
    \frac{p+p^{s-1}(1-2p)}{1-p}\right)
    \label{eq:shift_max_sum}\\
    &=\log\!\left[\frac{p}{1-p}+
    \frac{(1-2p)(1-p^N)}{N(1-p)^2}\right].
    \label{eq:shift_max}
\end{align}

\noindent\textit{(ii) MI leakage upper bound:}
\begin{equation}
    \varepsilon_{\mathrm{shift}}^{\mathrm{MI}}
    \;\leq\;
    \frac{1}{N}\sum_{s=2}^{N}
    \varepsilon_{v_s}^{\mathrm{MI}}
    \;=\;
    \frac{1}{N}\sum_{s=2}^{N}
    \bigl[H\!\left(Q_{v_s}^{\uparrow}\right) - (s-1)\,\Hb(p)\bigr].
    \label{eq:shift_mi}
\end{equation}
Both bounds equal zero at $p = 1/2$, consistent with strict privacy.
\end{theorem}

\begin{proof}
\textit{Identical marginal distributions and rate invariance.}
For any server $v_n$, any query value $q$, and any file index $\theta$, conditioning on the shift $R$:
\begin{align}
    P(Q_n=q\mid\theta)
    &=\sum_{r=1}^{N}P(R=r)\cdot P(Q_n=q\mid\theta,R=r) \notag\\
    &=\frac{1}{N}\sum_{r=1}^{N}P_{\sigma_r^{-1}(n)}(Q=q\mid\theta).
    \label{eq:sym_expand}
\end{align}
Since $\sigma_r$ is a cyclic permutation of $[N]$, as $r$ ranges over $[N]$ the position $\sigma_r^{-1}(n)$ takes each value in $[N]$ exactly once. Hence~\eqref{eq:sym_expand} equals $\frac{1}{N}\sum_{s=1}^{N}P_s(Q=q\mid\theta)$, independent of $n$, giving every server the same marginal distribution and hence identical leakage. For the download, since each $v_n$ occupies position $s$ with probability $1/N$:
\begin{align}
    \mathbb{E}[D_{\mathrm{shift}}]
    =\sum_{n=1}^{N}\frac{1}{N}\sum_{s=1}^{N}P_s(Q\neq\mathbf{0})
    =\sum_{s=1}^{N}P_s(Q\neq\mathbf{0})
    =\mathbb{E}[D_{\KN}(p)].\notag
\end{align}

\textit{Part~(i).}
Using $\max_\theta\frac{1}{N}\sum_s f_s(\theta)\leq
\frac{1}{N}\sum_s\max_\theta f_s(\theta)$ for non-negative $f_s$:
\begin{align}
    \varepsilon_{\mathrm{shift}}^{\mathrm{max}}
    &= \log\sum_{q}\max_{\theta}
       \frac{1}{N}\sum_{s=1}^{N} P_s(q\mid\theta)
    \notag\\
    &\leq
    \log\!\left(\frac{1}{N}\sum_{s=1}^{N}
    \underbrace{\sum_{q}\max_{\theta} P_s(q\mid\theta)}_{
    =\,\exp(\varepsilon_{v_s}^{\mathrm{max}})}\right).
    \label{eq:max_ub_step}
\end{align}
Substituting $\exp(\varepsilon_{v_s}^{\mathrm{max}})=[p+p^{s-1}(1-2p)]/(1-p)$ from~\eqref{eq:max_complete} gives~\eqref{eq:shift_max_sum}. Evaluating the geometric sum $\sum_{s=1}^{N}p^{s-1}(1-2p)=(1-2p)(1-p^N)/(1-p)$ and dividing by $N(1-p)$ gives~\eqref{eq:shift_max}.

\textit{Part~(ii).}
Since $R \perp \theta$, the chain rule gives
\begin{align}
    \varepsilon_{\mathrm{shift}}^{\mathrm{MI}}
    &= I(\theta;\, Q_n)
    \;\leq\; I(\theta;\, Q_n, R)
    \notag\\
    &= \underbrace{I(\theta;\, R)}_{=\,0}
       + I(\theta;\, Q_n \mid R)
    \notag\\
    &= \frac{1}{N}\sum_{s=1}^{N} I(\theta;\, Q_{v_s})
    = \frac{1}{N}\sum_{s=2}^{N} \varepsilon_{v_s}^{\mathrm{MI}},
    \label{eq:shift_mi_proof}
\end{align}
where the last equality uses $I(\theta; Q_{v_1}) = 0$.
\end{proof}
 
\begin{remark}
\label{rem:shift_gain}
Both bounds are strictly smaller than the worst-case per-server leakage of Protocol~1. For the maximal leakage, the bound~\eqref{eq:shift_max} satisfies
$\varepsilon_{\mathrm{shift}}^{\mathrm{max}} \leq \log\!\left(\frac{1}{N} \sum_s e^{\varepsilon_{v_s}^{\mathrm{max}}}\right) \leq \varepsilon_{v_N}^{\mathrm{max}}$,
where the second inequality holds because the arithmetic mean of $e^{\varepsilon_{v_s}^{\mathrm{max}}}$ is at most its maximum. The actual per-server leakage of the mixture is strictly below the bound; for example, at $N=3$ and $p=0.7$, the exact mixture maximal leakage is $0.125$, the bound~\eqref{eq:shift_max} gives $0.307$, and the worst-case leakage under Protocol~1 is $\varepsilon_{v_3}^{\mathrm{max}} = 0.452$. These improvements are visible in Fig.~\ref{fig:comlete_graph_MaxL}, where the dotted line (Protocol~2) lies strictly below the solid curves (Protocol~1) for all $p > 1/2$. 
\end{remark}

%\pk{Say something about Fig. \ref{fig:comlete_graph_MaxL} here. }

\begin{figure}
    \centering
    \includegraphics[width=1\linewidth]{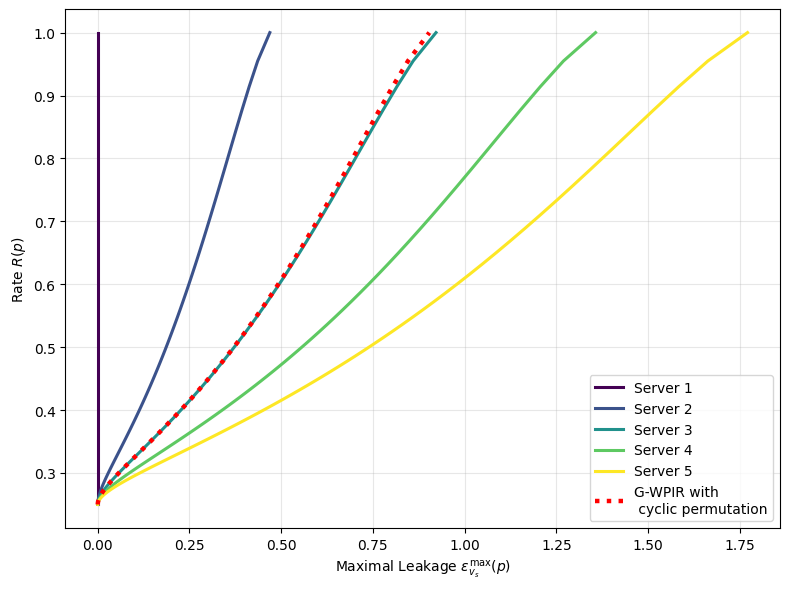}
    \caption{Rate and $\varepsilon_{v_s}^{\mathrm{MI}}$ (MI leakage) trade-offs for complete graph with $5$ vertices ($K_5$), at each server. The dotted line represents the tradeoff achieved at each server by the cyclic-shift-based protocol in Subsection \ref{subsec:cyclic_shift}.}
    % \pk{PLEASE SYNCHRONIZE NOTATIONS WITH THIS PAPER. $\mileak$ IS NOT USED. SIMILARLY $\{0,1,\hdots,N-1\}$ is not used, it is $\{1,\hdots,N\}$ I BELIEVE! INCREASE THE FONT SIZE FOR THE AXES AND LEGENDS.
    % DO IT FOR OTHER CURVE ALSO. PUT THE CURVES ON/AFTER THE PAGE THEY ARE BEING TALKED ABOUT IN THE MAIN TEXT. OTHERWISE AUDIENCE WONDERS WHY ITS HERE. }
    % } 
    \label{fig:comlete_graph_MIL}
\end{figure}

\begin{figure}
    \centering
    \includegraphics[width=1\linewidth]{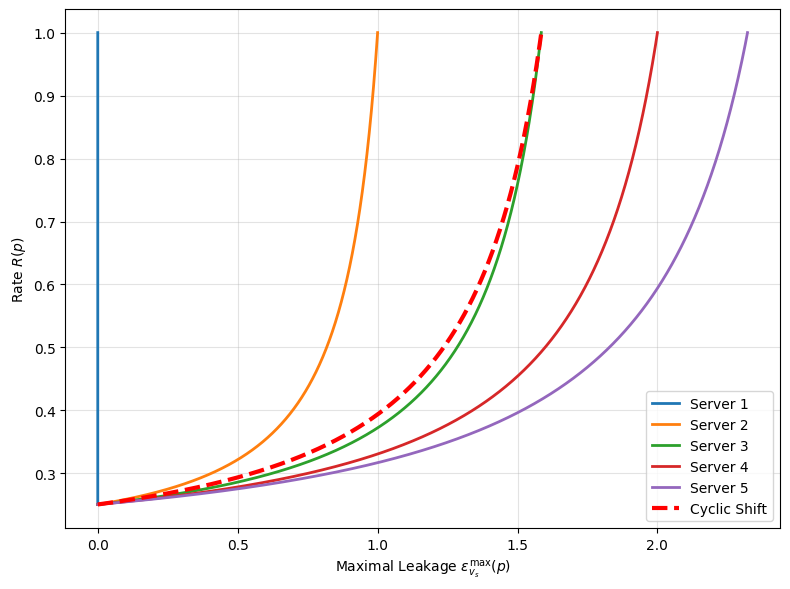}
    \caption{Rate and $\varepsilon_{v_s}^{\mathrm{max}}$ (maximal leakage) trade-offs for complete graph with $5$ vertices ($K_5$). The dotted line represents the tradeoff achieved at each server by the cyclic-shift-based protocol in Subsection \ref{subsec:cyclic_shift}.}
    \label{fig:comlete_graph_MaxL}
\end{figure}

\subsection{Complete Bipartite Graphs}
\label{subsec:bipartite}
 
We now specialise our general protocol in Section \ref{subsec:scheme} to the complete bipartite graph $K_{N_1, N_2}$, where side $A$ has $N_1$ vertices and side $B$ has $N_2$ vertices. Assume without loss of generality that $N_1 \geq N_2$.  The total number of servers is $N = N_1 + N_2$,
the total number of files is $K = N_1 N_2$, and the independence number is $\alpha(K_{N_1,N_2}) = N_1$.
 
\subsubsection{Partition Structure}
\label{subsubsec:bip_partition}
 
The sequential independent-set partition of Section~\ref{subsec:scheme} produces exactly $\kappa = 2$ sets: $\indset{1} = A$ (side $A$, $N_1$ vertices) and $\indset{2} = B$ (side $B$, $N_2$ vertices).  Since $K_{N_1,N_2}$ is bipartite, every edge crosses between the two sides, so all edges incident on $n \in \indset{1}$ are downstream and all edges incident on $n \in \indset{2}$ are upstream.  Concretely:
\begin{itemize}
  \item Every server $n \in \indset{1}$ has $\degdown = N_2$, $\degup = 0$,
        $\Un = 0$.
  \item Every server $n \in \indset{2}$ has $\degdown = 0$, $\degup = N_1$,
        $\Un = N_1$.
\end{itemize}
In particular, all $N_1$ servers on side $A$ are structurally identical (zero upstream edges, and $N_2$ downstream edges) and all $N_2$ servers on side $B$ are structurally identical ($N_1$ upstream edges each, and no downstream edges), so the leakage and download contributions depend only on which side a server belongs to.
 
\subsubsection{Leakage Expressions}
\label{subsubsec:bip_leakage}
 
\begin{theorem}[Leakage-Rate Tradeoff for $K_{N_1,N_2}$]
\label{thm:bip_tradeoff}
The biased protocol on $K_{N_1,N_2}$ with parameter $p \in [1/2,\, 1]$ satisfies:

\noindent\textit{(i) Side-$A$ servers ($n \in \indset{1}$):}
Every side-$A$ server has zero leakage, i.e.,
\begin{equation}
    \varepsilon_n^{\mathrm{MI}} \;=\; 0, \qquad
    \varepsilon_n^{\mathrm{max}} \;=\; 0.
    \label{eq:bip_leak_A}
\end{equation}

\noindent\textit{(ii) Side-$B$ servers ($n \in \indset{2}$), all equal:}
For each side-$B$ server $n$, we have the leakage metrics as follows
\begin{align}
    \varepsilon_n^{\mathrm{max}}
    &\;=\; \log\!\left[\frac{p + p^{N_1}(1-2p)}{1-p}\right],
    \label{eq:bip_max}\\
    \varepsilon_n^{\mathrm{MI}}
    &\;=\; H\!\left(\queryup{n}\right) - N_1\,\Hb(p),
    \label{eq:bip_mi}
\end{align}
where $H(\queryup{n})$ is given by~\eqref{eq:H_up} and~\eqref{eq:marginal}
with $\Un = N_1$ and $K = N_1 N_2$.

\noindent\textit{(iii) Expected download and rate:}
\begin{equation}
    \mathbb{E}\!\left[D_{K_{N_1,N_2}}(p)\right]
    \;=\; N_1(1-p) \;+\; N_2 \;-\; p^{N_1-1}\bigl[1 + p(N_2-2)\bigr],
    \label{eq:bip_download}
\end{equation}
and rate $R(p) = 1/\mathbb{E}[D_{K_{N_1,N_2}}(p)]$.
\end{theorem}

\begin{proof}
Part~\textit{(i)} follows directly from Theorems~\ref{thm:mi_rate} and~\ref{thm:maxl_rate}, since side-$A$ servers lie in $\indset{1}$, where both theorems assign zero
leakage irrespective of $\Un$. Part~\textit{(ii)} follows by substituting $\Un = N_1$ into Theorems~\ref{thm:mi_rate} and~\ref{thm:maxl_rate}, applied to side-$B$
servers, which lie in $\indset{2} = \indset{\kappa}$. Equality of leakage across all side-$B$ servers follows from their structural identity: each has
exactly $N_1$ upstream edges and $K = N_1 N_2$.

Part~\textit{(iii)}: each side-$A$ server contributes $P(Q_n \neq \mathbf{0}) = 1-p$.
For a side-$B$ server, we have ($\degdown = 0$, $\Un = N_1$). Substituting all these
into~\eqref{eq:null_last}:
\begin{align}
    P\!\left(Q_n = \mathbf{0}\right)
    &= \frac{p^{N_1}}{N_1 N_2}\!\left[N_1 N_2 - \frac{N_1(2p-1)}{p}\right] \\
    &= \frac{p^{N_1-1}\bigl[1 + p(N_2-2)\bigr]}{N_2}.
    \label{eq:bip_null_B}
\end{align}
Summing over all $N_1 + N_2$ servers gives~\eqref{eq:bip_download}.
\end{proof}

\begin{remark}
The expected download~\eqref{eq:bip_download} satisfies
\begin{equation}
    \mathbb{E}\!\left[D_{K_{N_1,N_2}}\!\left(\tfrac{1}{2}\right)\right]
    = \frac{N_1}{2} + N_2 - \frac{N_2}{2^{N_1}},
    \qquad
    \mathbb{E}\!\left[D_{K_{N_1,N_2}}(1)\right] = 1,
    \label{eq:bip_boundary}
\end{equation}
giving $R(1/2) = (N_1/2 + N_2 - N_2/2^{N_1})^{-1}$ and $R(1) = 1$.
At $p=1/2$: $(2p-1)=0$, so $P(Q_n=\mathbf{0})=(1/2)^{N_1}$ per
side-$B$ server; using $1+(1/2)(N_2-2)=N_2/2$ gives the first expression.
At $p=1$, side-$A$ adds nothing to the total download, each side-$B$ server has
$P(Q_n \neq \mathbf{0})=1/N_2$, giving $\mathbb{E}[D]=1$.
\end{remark}

\section{Conclusion}
\label{sec:conclusion}
In this work, we have initiated the study of rate-privacy tradeoffs for PIR on graph-based storage. Clearly, multiple directions of further work exist, including obtaining information-theoretic converses for WPIR on graphs, improved schemes, and schemes for special graphs. Further, PIR for collusion of servers in graph-based storage and other generalizations remains largely open. 

\bibliographystyle{IEEEtran}
\bibliography{ref.bib}

\end{document}